\documentclass[conference]{IEEEtran}
\IEEEoverridecommandlockouts
\usepackage{graphicx} 
\usepackage{hyperref}
\usepackage{amsmath}
\usepackage{amssymb}
\usepackage{qcircuit}
\usepackage[dvipsnames]{xcolor}
\usepackage{braket}
\usepackage{amsthm}
\usepackage{algorithm}
\usepackage{multirow}
\usepackage{booktabs}
\usepackage{subcaption}
\usepackage{algpseudocode}
\usepackage{tikz}

\newcommand{\norme}[1]{\lVert #1 \rVert}

\newcommand{\tx}{\Tilde{x}}

\newcommand{\ep}{\epsilon}
\newcommand{\ka}{\kappa}
\newcommand{\norm}[1]{\Vert {#1} \Vert}
\newcommand{\bigo}{\mathcal{O}}

\newtheorem{theorem}{Theorem}[section]

\newtheorem{remark}{Remark}

\begin{document}

\title{A mixed-precision quantum-classical algorithm for solving linear systems}

\author{\IEEEauthorblockN{1\textsuperscript{st} Océane Koska}
\IEEEauthorblockA{\textit{Université Paris-Saclay and Eviden} \\
Orsay, France \\
oceane.koska@eviden.com }
\and
\IEEEauthorblockN{2\textsuperscript{nd} Marc Baboulin}
\IEEEauthorblockA{\textit{Université Paris-Saclay and Inria}\\
Orsay, France \\
marc.baboulin@inria.fr}
\and
\IEEEauthorblockN{3\textsuperscript{rd} Arnaud Gazda}
\IEEEauthorblockA{\textit{Eviden Quantum Lab} \\
Les Clayes-sous-Bois, France\\
arnaud.gazda@eviden.com}
}

\maketitle

\begin{abstract}
We address the problem of solving a system of linear equations via the Quantum Singular Value Transformation (QSVT). One drawback of the QSVT algorithm is that it requires huge quantum resources if we want to achieve an acceptable accuracy. To reduce the quantum cost, we propose a hybrid quantum-classical algorithm that improves the accuracy and reduces the cost of the QSVT by adding iterative refinement in mixed-precision A first quantum solution is computed using the QSVT, in low precision, and then refined in higher precision until we get a satisfactory accuracy.
For this solver, we present an error and complexity analysis, and first experiments using the quantum software stack myQLM.
\end{abstract}


\section{Introduction}
In this paper we address the linear system (LS) problem where, given a nonsingular matrix $A \in \mathbb{R}^{N \times N}$ and a vector $b \in \mathbb{R}^N$, we want to compute $x \in \mathbb{R}^N$ such that
\begin{equation}\label{eq:LS}
Ax = b.    
\end{equation}
Solving the LS problem accurately and efficiently is a fundamental problem in computational science for which there exist many algorithms and software libraries on classical processors~\cite{LAPACK, GVL13}. Numerical methods based on factorization or iterative algorithms have reached a high level of maturity on classical computers in particular with the impressive performance of 1 Exaflop/s achieved for solving a dense linear system by Gaussian elimination (Linpack benchmark~\cite{LINPACK}). The perspective of having operational quantum computers in a near future has motivated the development of quantum algorithms for the LS problem, with the promise of either obtaining faster solution or solving problems that are currently intractable with classical supercomputers. With algorithms such as Harrow-Hassidim-Lloyd (HHL) \cite{HHL09}, the Quantum Singular Value Transformation (QSVT) \cite{GSLW19},  or the Variational Quantum Linear Solver (VQLS) \cite{VQLS}, quantum algorithms for linear systems offer the potential for exponential speedups under specific conditions.

However, these quantum algorithms currently have limitations in terms of matrix conditioning and solution accuracy. These can be mitigated by using a combination of both classical computations (running on a Central Processing Unit - CPU) and quantum computations (running on a Quantum Processing Unit - QPU). Recent studies have proposed architectural designs for the integration of quantum devices into High-Performance Computing (HPC) systems, to improve the computation accuracy and to reduce the latency \cite{BSH20, HMLG21}. This type of architectures will be beneficial for implementations that require data transfer between CPU and QPU. At the algorithm level, the techniques that could be targeted to improve linear system solvers could for instance include preconditioning methods or iterative refinement.

Iterative refinement has been recently studied for the HHL algorithm \cite{IR_HHL_saito2021, IR_HHL_zheng2024} or for optimization problems \cite{IR_opti}.
In this paper we consider the LS problem solved via the QSVT method which offers several advantages: the matrix does not need to be Hermitian and even not square (in this case we solve a least squares problem), the QSVT method can exploit efficient block-encoding techniques and appears to be well suited for Large Scale Quantum (LSQ) architectures. Note that our work is carried out in an LSQ context and not NISQ (Noisy Intermediate-Scale Quantum) due to the excessive depth of quantum circuits for the QSVT algorithm. The quantum circuit for the QSVT requires a polynomial approximation that is very expensive if we want to achieve acceptable accuracy (typically better than $10^{-5}$ \cite{NJ24}). This has motivated our choice for developing a mixed-precision algorithm which combines the speed and limited precision of the QSVT with iterative refinement in higher precision. 
Mixed-precision techniques are widely used in HPC algorithms \cite{Mixed_2016}. They have been initially motivated by the emergence of GPU accelerators that provide high performance when computing in  lower precision arithmetic, for instance using tensor core units \cite{Baboulin_Europar}. For hybrid CPU/GPU architectures we can achieve high performance of the solver while preserving the accuracy of the higher precision and we investigate how similar techniques can be applied with the QSVT method. Similarly to the CPU/GPU case, our iterative refinement is also hybrid since it uses a classical processor (to compute the residual and solution update) and a quantum processor (for the QSVT solver). We adapt the iterative refinement to the case where the LS problem is solved via the QSVT and we present an error and complexity analysis to evaluate the quantum ans classical costs of our algorithm. We also show that our approach provides a significant advantage in complexity compared to directly solving the LS problem with QSVT in higher precision.

The paper is organized as follows: in Section \ref{sec:background}, we recall some results on  QSVT and classical mixed-precision iterative refinement applied to LS problems. Then in Section \ref{sec:iterref}, we describe a hybrid algorithm for LS problems using iterative refinement based on the QSVT. We also provide a convergence and complexity analysis. In Section \ref{sec:exp}, we present experimental results on random matrices with various condition numbers. Finally, concluding remarks are given in Section \ref{sec:conclusion}.

Notation: Unless otherwise stated, $\norm{.}$ denotes the Euclidean norm for vectors and the spectral norm for matrices. 


\section{Background} \label{sec:background}

\subsection{Quantum Singular Value Transformation (QSVT)}
\subsubsection{Block-encoding of matrices}\label{sec:block_encoding}
Since quantum algorithms can only handle unitary matrices,
the block-encoding technique consists in embedding a non-unitary matrix into a unitary one \cite{CLVY23, CV22}. Namely a general matrix $A \in \mathbb{C}^{N\times N}$ (with $N=2^n$) is encoded into a unitary matrix $U$ as 
$$
    U = 
    \begin{bmatrix}
    A & \cdot \\
    \cdot & \cdot
\end{bmatrix}.
$$
The matrix $A$ can be expressed from the unitary $U$ using two projectors $\Tilde{\Pi}$ and $\Pi$ with
$$
A = \Tilde{\Pi} U \Pi.
$$
When we apply $U$ to $\ket{0}_a\ket{\psi}_d$ , where $\ket{0}_a$ corresponds to the ancilla qubits and $\ket{\psi}_d$ corresponds to the data qubits, we get

$$
    U \left( \ket{0}_a \ket{\psi}_d \right)= \ket{0}_a A \ket{\psi}_d + \cdots.
$$

The literature provides several block-encoding methods to encode an arbitrary matrix into a unitary. The Linear Combination of Unitaries (LCU) method is a versatile approach to encode matrices~\cite{CW12}. It consists in representing the matrix $A$ as a weighted sum of unitary operators $A = \sum_j \alpha_j U_j$.
This method relies on the state preparation of the coefficients $\alpha_j$ using ancillary qubits and controlled operations. This method can be used for general matrices~\cite{KBG24}. The Fast Approximate Block Encoding (FABLE) algorithm~\cite{CV22} provides an efficient way to construct an approximation of block-encoding by eliminating the negligible terms and removing useless controls in the LCU approach. This algorithm achieves a complexity bounded by $\mathcal{O}(4^n)$ gates for general and unstructured matrices. However it is also possible to take advantage of the sparsity and the structure of some matrices to reach a complexity in $\mathcal{O}(poly(n))$ \cite{CLVY23}.

\subsubsection{QSVT definition}

The Quantum Singular Value Transformation (QSVT) is a matrix function that operates on the singular values of a matrix using a quantum computer \cite{GSLW19}. Given a matrix $A \in \mathbb{C}^{N \times N}$ with the Singular Value Decomposition (SVD)
$$
A = W \Sigma V^\dagger
$$
($\Sigma=diag(\sigma_1,\dots,\sigma_N)$ and  $W,~V$ unitary matrices), given a polynomial $P$ of degree $d$, then the QSVT of $A$ using $P$ is expressed by 
$$
    W \Sigma V^\dagger =
    \begin{cases}
        W P(\Sigma) V^\dagger \; \mathrm{if} \; d \; \mathrm{is\; odd,}\\
        V P(\Sigma) V^\dagger \; \mathrm{if} \; d \; \mathrm{is\; even.}
    \end{cases}
$$

However, the QSVT imposes some constraints on the polynomial used in the transformation \cite{GSLW19, LYC16}:

\begin{itemize}
    \label{itemize:QSP_poly}
    \item $P$ has parity-$(d \; \mathrm{mod} \; 2)$,
    \item $\forall x \in [-1,1]$, $|P(x)| \leq 1$,
    \item $\forall x \in (-\infty, -1] \cup [1, \infty)$, $|P(x)| \geq 1$,
    \item if $d$ is even, then $\forall x \in \mathbb{R}$, $P(ix)P^*(ix) \geq 1$.
\end{itemize}

\subsubsection{QSVT implementation}\label{sec:QSVT_implem}
If the condition above 
are respected, there exists a vector of phases $\Phi = \{\phi_1, \dots, \phi_d\} \in \mathbb{R}^d$ such that we can define an alternating phase modulation sequence operator $U_\Phi$ such that:\\
if $d$ is odd,
\begin{equation}
    \label{eq:apms_operator_odd}
    U_\Phi = e^{i\phi_1(2\tilde\Pi-I)}U\prod_{j=1}^{(d-1)/2}\left( e^{i\phi_{2j}(2\Pi-I)}U^\dagger e^{i\phi_{2j+1}(2\tilde\Pi-I)}U\right)
\end{equation}
and if $d$ is even,
\begin{equation}
    \label{eq:apms_operator_even}
    U_\Phi = \prod_{j=1}^{d/2}\left( e^{i\phi_{2j-1}(2\Pi-I)}U^\dagger e^{i\phi_{2j}(2\tilde\Pi-I)}U\right).
\end{equation}

This operator can be used to apply the QSVT in a quantum computer, according to the polynomial $P$, to the matrix $A$
$$
    QSVT^P(A) =
    \begin{cases}
        \tilde\Pi U_\Phi \Pi, \mathrm{if} \; d \; \mathrm{is \; odd}\\
        \Pi U_\Phi \Pi, \mathrm{if} \; d \; \mathrm{is \; even}
    \end{cases}
$$

\begin{remark}\label{rem:QSVT_complexity}
    For a problem of size $2^n$ with a polynomial of degree $d$, then it comes from~\cite{GSLW19}) that the alternating phase modulation sequence described in Equations (\ref{eq:apms_operator_odd}) and (\ref{eq:apms_operator_even}) can be efficiently implemented using
    \begin{itemize}
        \item $n$ data qubits,
        \item  a single ancilla qubit,
        \item $d$ calls to the block-encoding $U$ and $U^\dagger$,
        \item $d$ calls to the operators of the form $e^{i\phi(2\Pi - I)}$ and $e^{i\phi(2\tilde\Pi - I)}$.
    \end{itemize}
Then the circuit's depth and complexity scale logarithmically with the problem dimension and linearly with the degree of the polynomial.
\end{remark}

\subsubsection{Solving linear systems with QSVT}\label{sec:qsvt4ls}
To solve LS problems via the QSVT we need to find an odd-degree polynomial approximation of the inverse function that fits the requirements given in Section \ref{itemize:QSP_poly}. By applying this polynomial $P$ for the QSVT of the matrix $A^\dagger$ we get:
$$
    QSVT^{P}(A^\dagger) = V P(\Sigma)W^\dagger \approx V \Sigma^{-1} W^\dagger = A^{-1},
$$
where $\Sigma^{-1}=diag(\sigma_1^{-1},\dots,\sigma_N^{-1})$.

In practice we want to find an $\frac{\epsilon}{2\kappa}$-approximation to $\frac{1}{2\kappa} \frac{1}{x}$, where $\kappa$ is the condition number of the matrix to be inverted, and $\epsilon$ the error on $[-1,1] \setminus [\frac{-1}{\kappa}, \frac{1}{\kappa}]$. The inverse function is difficult to approximate with a polynomial (not continuous, infinite values in 0...). We need to find an odd function approximating the inverse function on $[-1,1] \setminus [\frac{-1}{\kappa}, \frac{1}{\kappa}]$ that would be easier to approximate using a polynomial. A function that can be used is 
$$
f_{\epsilon, \kappa}(x) = \frac{1 - (1-x^2)^b}{x},
$$
where $b(\epsilon, \kappa) = \lceil \kappa^2 log(\kappa / \epsilon)\rceil$ \cite{GSLW19}.

Then this function can be $\epsilon$-approximated by the polynomial
\begin{equation}
\label{eq:poly}
    P^{1/x}_{2\epsilon, \kappa}(x) = 4\sum_{j=0}^D (-1)^j \left[2^{-2b} \sum_{i=j+1}^b \binom{b+i}{2b}\right] T_{2j+1}(x),
\end{equation}

where $T_i(x)$ is the Chebyshev polynomial of first kind of order $i$, and $D(\epsilon, \kappa) = \lceil \sqrt{b(\epsilon, \kappa) \; log(4b(\epsilon, \kappa)/\epsilon)} \rceil$ \cite{MRTC21}. Using Chebyshev polynomials to perform the polynomial approximation (instead of expressing the approximation in the canonical basis) highly reduces the impact of Runge's phenomenon when working with high degree polynomials \cite{R20}.

However, this polynomial does not fit the conditions of the QSVT, because it is not necessarily bounded in magnitude by $1$ for $x \in \left[\frac{-1}{2\kappa}, \frac{1}{2\kappa} \right]$. To enforce the magnitude to be bounded, we need to multiply this polynomial by another one. A polynomial that approximates a rectangular function will fit these constraints \cite{MRTC21}.

One of the main challenges in solving linear systems using QSVT lies in accurately approximating the inverse function with a polynomial, without drastically increasing the polynomial's degree. This constraint impacts the achievable accuracy in linear system solving and/or limits the maximum condition number that can be addressed. Current state-of-the-art methods achieve a precision of $10^{-5}$ for condition numbers as large as $10^6$ \cite{NJ24}. A summary of error analysis for the polynomial approximation and the resulting forward error on the LS solution can be found in ~\cite[p. 125]{L22}.

\subsection{Mixed-precision iterative refinement for linear systems}\label{sec:itergeneral}

Iterative refinement~\cite{H96,W63} is a well-known technique that improves a computed solution $\tx$ to $Ax=b$ (step 0) by performing the following steps:
\begin{enumerate}
    \item compute the residual $r=b-A\tx$,
    \item solve $Ae=r$ to obtain the correction vector $e$,
    \item update $\tx \leftarrow \tx + e$ which gives a $\tx$ ``closer'' to the exact solution $x$.    \end{enumerate}
This process can then be repeated until we obtain a satisfying $\tx$. Using the same precision arithmetic throughout the process (\textit{fixed-precision} iterative refinement) is classically used to improve the accuracy of a computed solution or added to ameliorate a potentially instable solver~\cite[p. 232]{H96}.
However, with the emergence of processors that propose much faster computation using lower precision arithmetic, it became attractive to combine different precisions in order to exploit the high performance provided by some processors or computational units (e.g., tensor cores for NVIDIA GPUs), resulting in so-called \textit{mixed-precision} iterative refinement algorithms.

When solving linear systems, mixed-precision is of interest when the refinement (computed in higher precision) is cheap compared to the computation of the first $\tx$ (computed in lower precision). In mixed-precision iterative refinement for general matrices, the most expensive tasks are the initial solution of $Ax=b$ (step 0) and the successive solves of $Ae=r$ (step 2) which are then achieved in lower precision. On the other hand the residual and the updated $\tx$ (steps 1 and 3) are computed in higher precision. Algorithm \ref{alg:iter_ref} gives the precisions that are used for each task of the algorithm.
When the solves are performed using LU factorization, then we can use the L and U factors produced in step 0 to achieve step 2, which still reduces the computational cost of the overall process.

A general framework that describes mixed-precision algorithms for linear systems can be found in
 \cite{CH18}. A special case can be derived that corresponds to a common use in heterogeneous computing (see e.g., \cite{BBDK09}) where we have a working (high) precision $u$ (e.g., double precision, $u=10^{-16}$) and a low precision $u_l$ (e.g., single precision, $u=10^{-8}$). In this situation the most expensive part is performed at precision $u_l$ on an accelerator like a GPU to take advantage of low precision arithmetic while the
 other steps remain executed on the CPU.
  
\begin{algorithm}
\caption{Mixed-precision linear system solution using 2 precisions.}\label{alg:iter_ref}
\begin{algorithmic}
\State \textbf{Input}:  $A \in \mathbb{R}^{N\times N}, b \in \mathbb{R}^{N}$ stored at precision $u$ with $u \ll u_l$.
\State Compute a solution $x_0$ to $Ax=b$ at precision $u_l$.
\While{desired accuracy not reached}
\State Compute $r_i = b - Ax_i$ at precision $u$.
\State Solve $A e_i = r_i$ at precision $u_l$.
\State Update $x_{i+1} = x_i + e_i$ at precision $u$.
\EndWhile
\end{algorithmic}
\end{algorithm}

The limiting accuracy does not depend on $u_l$ in the system solving $Ae_i = r_i$, it only depends on the choice for $u$~\cite{HM22}. Therefore we can work with $u \ll u_l$ and still get an accurate result if $u$ is well chosen. Traditionally, we have $u = u_l^2$ to get a limiting accuracy of order $u$. Note that a large value of $u_l$ accelerates each iteration but requires more iterations. 

\section{Iterative refinement for QSVT-based linear system solution}\label{sec:iterref}

\subsection{Algorithm}
Suppose we can solve $Ax=b$ via the QSVT method with a (low) accuracy $\ep_l$, i.e., that can produce a solution $\tx$ such that $\norm{x-\tx} \leq \ep_l \norme{x}$. Following \cite[p. 126]{L22}, to obtain an accuracy (relative error) of order $\ep_l$ on the non-normalized solution $\tx$, we need to approximate the inverse function on $[-1, -1/ \kappa]\cup[1/\ka,1]$ with an error $\ep' =\bigo(\ep_l / \ka)$.

We want to improve the quality of the solution given by the QSVT described in Section:\ref{sec:itergeneral} by refining it in higher precision $u$ as presented in 
Algorithm \ref{alg:iter_ref_QSVT} until we achieve an accuracy $\ep$.
In this algorithm we will use two different types of processors: the solving phases via the QSVT will be achieved on a quantum processor (QPU) while the residual $r_i$ and the correction to the solution $x_{i+1}$ will be computed on a classical processor (CPU). Note that in practice the QSVT routine, like most quantum algorithms, is by nature hybrid and not fully executed on the QPU. Indeed some tasks are performed on the CPU (matrix decomposition before block-encoding, preprocessing of the state preparation, computation of angles for the polynomial approximation, post-processing after measurement). 
Note also that the hybrid scheme of Algorithm \ref{alg:iter_ref_QSVT} requires to store $A$ and $b$ in a precision at least $u$ on the CPU. 

\begin{remark}\label{rem:normalization}
One particularity of using iterative refinement in quantum algorithms is that we deal with quantum states. Therefore before solving $Ax = b$, we need to normalize $b$ as
$$
A \frac{x}{\norme{b}} = \frac{b}{\norme{b}}.
$$
The sampling at the end of the QSVT will provide $\eta=\frac{x}{\norme{x}}$ and we are able to recover $\norme{x}$ by solving the minimization problem
$$
argmin_{\mu\in \mathbb{R}} \lvert A(x+\mu \eta)- b \rvert.
$$    
This phase occurs for each call to the QSVT routine and is performed on the CPU device.
\end{remark}

The \textbf{stopping criterion} for our iterative refinement will be based on the scaled residual defined 
by $\omega= \frac{\norm{b-A\tx}}{\norm{b}}$. We aim at finding $\tx$ such that $\omega \leq \ep$.
When $\ka$ is not too large, $\omega$ classically provides close bounds for the relative error since we have (see, e.g.,~\cite[p. 68]{LAUB})
\begin{equation}\label{residual_error}
\frac{\omega}{\ka} \leq \frac{\norm{x-\tx}}{\norm{x}} \leq \ka \omega.
\end{equation}
Moreover $\omega$ is independent to scaling of $Ax$ and $b$ by a same coefficient, which will occur because quantum algorithms require $b$ to be normalized (see Remark \ref{rem:normalization}).
 
\begin{algorithm}[ht!]
\caption{Iterative refinement for QSVT-based linear system solution.}\label{alg:iter_ref_QSVT}
\begin{algorithmic}
\State \textbf{Input}:  $A$, $b$, QSVT accuracy $\ep_l$, targeted accuracy $\ep$ at precision $u$.
\State Compute $x_0=A^{-1}b$ at accuracy $\ep_l$ using QSVT \textcolor{red}{(QPU)}.
\While{accuracy $\ep$ is not reached on $x_i$}
\State Compute $r_i = b - Ax_i$ in high precision $u$ \textcolor{blue}{(CPU)}.
\State Compute $e_i = A^{-1}r_i$ at accuracy $\ep_l$ with QSVT \textcolor{red}{(QPU)}.
\State Update $x_{i+1} = x_i + e_i$ in high precision $u$ \textcolor{blue}{(CPU)}.
\EndWhile
\end{algorithmic}
\end{algorithm}

To compute $x_0$ we first need to generate all the quantum circuits/routines that will be executed on the QPU. This generation is called quantum circuit synthesis, which corresponds in classical computing to the compilation phase, and is executed on a classical computer. Computing $x_0$ requires the use of 3 quantum routines:
\begin{itemize}
    \item State preparation implementation for the normalized value of $b$,
    \item Block-encoding of $A^\dagger$,
    \item QSVT routine implementing $A^{-1}$ (relying on the block encoding of $A^{\dagger}$).
\end{itemize}

Quantum routines, once compiled and transferred to the QPU, do not have to be redefined. In particular, the concept of ``linker-loader", widely used in classical computing, is part of the architecture for integrating QPUs into HPC resources (see, e.g., Figure 4 in \cite{BSH20}). Thus, each iteration of Algorithm~\ref{alg:iter_ref_QSVT} requires to transfer a reduced amount of data compared to computing $x_0$ (since only $r_i$ needs to be encoded and transferred to the QPU).


\begin{remark}
We point out that the result is obtained through sampling/measurement. Consequently, this hybrid algorithm relies on the ``collapse'' of the quantum solution and cannot be used in a subsequent quantum algorithm, except if we re-encode the solution in the quantum computer via state preparation.
\end{remark}

\subsection{Convergence and accuracy}
In this section we compute the scaled residual obtained after each iteration of the refinement and a bound on the iteration count. In our demonstrations we omit the effect of rounding errors but the high precision $u$ (unit roundoff) used for CPU computations should be chosen accordingly to the target precision $\ep$, a safe choice being 
$u=\theta \ep$ with $\theta \leq 1$.



\begin{theorem}\label{theo:residual}
Suppose the QSVT can solve $Ax=b$ with low accuracy $\ep_l$ with $\ep_l \ka < 1 $ and that 
we apply mixed-precision iterative refinement as given in Algorithm~\ref{alg:iter_ref_QSVT} with high precision $u$ when computing the residual and solution update.
Then after $i$ iterations we have $\norm{r_i} \leq (\ep_l \ka)^{i+1} \norm{b}$.
The number of iterations to obtain a solution $\tx$ such that 
$\frac{\norm{b-A\tx}}{\norm{b}} \leq \ep$
is bounded by 
$\lceil\log(\ep) / \log(\ep_l \ka)\rceil$.
\end{theorem}

\begin{proof}
We first compute $x_0$ with the QSVT with $\norm{x-x_0} \leq \ep_l \norm{x_0}$. Then using the left part of Equation (\ref{residual_error}) we have $$\norm{r_0}=\norm{b-Ax_0} \leq \ep_l \ka \norm{b}.$$ 
Then for the first iteration we have
$$
\norm{r_1}=\norm{b-Ax_1}
=\norm{b-A(x_0+e_0)}
=\norm{r_0 - Ae_0}.
$$
Using again Equation (\ref{residual_error}) to the linear system $Ae=r_0$ we get
$$\norm{r_1} = \norm{r_0 - Ae_0} \leq \ep_l \ka \norm{r_0} \leq (\ep_l\ka)^2 \norm{b}.
$$
Then by straightforward recurrence on $i$ we obtain 
$$\forall i, ~\norm{r_i} \leq (\ep_l \ka)^{i+1} \norm{b}.$$ 
The desired final error $\ep$ will be obtained when we will have
$$
\frac{\norm{b-A x_i}}{\norm{b}} \leq (\ep_l\ka)^{i+1} \leq \ep,  
$$
which yields $i\leq \log(\ep) / \log(\ep_l \ka) - 1$ and thus $\lceil\log(\ep) / \log(\ep_l \ka)\rceil$ will be an upper bound for the iterative refinement process.

\end{proof}

Theorem \ref{theo:residual} states that the convergence of Algorithm~\ref{alg:iter_ref_QSVT} is ensured as long as $\ep_l \ka < 1$. The scaled residual contracts by a factor $\ep_l \ka$ at each iteration until it reaches a maximum value of $\bigo(\ep)$ after at most $i_{max} = \lceil\log(\ep) / \log(\ep_l \ka)\rceil$ iterations. The quantity $i_{max}$ will be used in the following section to evaluate the complexity of the linear system solver. 

\subsection{Complexity analysis}
Since Algorithm \ref{alg:iter_ref_QSVT} is hybrid, the resulting complexity includes quantum and classical costs which are presented in this section. We also mention the data communication between the 2 devices.
\subsubsection{Quantum cost}
The quantum complexity of the QSVT (denoted by $\mathcal{C}_{QSVT}$ and detailed in Remark \ref{rem:QSVT_complexity}) mainly relies on the cost $\mathcal{B}$ of the block-encoding circuit used to encode $A^\dagger$, which depends on the chosen block-encoding method (see Section \ref{sec:block_encoding}).
Then the quantum complexity of the mixed-precision QSVT solver depends on the 3 following components: 
\begin{itemize}
    \item The \textbf{number of calls to the solver} is evaluated using the upper bound provided in Theorem \ref{theo:residual} for the mixed-precision solver.
    \item The \textbf{complexity of the QSVT} $\mathcal{C}_{QSVT}$, which corresponds to the cost of the multiple calls to the block-encoding. We recall that the number of calls to the block-encoding is given by the degree of the polynomial as explained in Remark \ref{rem:QSVT_complexity}.
    \item The \textbf{number of samples} necessary to achieved a targeted error $\epsilon$ is $\mathcal{O}(1/\epsilon^2)$.
\end{itemize}
Then we have
$$
\mathrm{Total \; complexity} \; =\;  \# \mathrm{solves} \; \times \; \mathcal{C}_{QSVT} \; \times \; \# \mathrm{samples}.
$$

In Table \ref{tab:complexity}, we compare the complexity when using directly QSVT in high precision (one call to the QSVT and $\ep_l=\ep$) and when using iterative refinement in mixed-precision.

\begin{table}[ht]
    \centering
    \begin{tabular}{lcc}
         & QSVT only & QSVT with iterative refinement \\ 
         \midrule
        \# solves & $1$ & $ \leq \left\lceil\dfrac{log(\epsilon)}{log(\kappa \; \epsilon_{l})}\right\rceil$ \\
        \midrule
        $\mathcal{C}_{QSVT}$ & $\mathcal{O}\left(\mathcal{B}\;\kappa \; log(\kappa / \epsilon\right))$ & $\mathcal{O}\left(\mathcal{B}\;\kappa\; log(\kappa / \epsilon_{l}\right))$\\
        \midrule
        \# samples & $\mathcal{O}(1/\epsilon^2)$ & $\mathcal{O}(1/\epsilon_{l}^2)$\\ 
        \midrule
        Total & $\mathcal{O}\left(\dfrac{\mathcal{B}\kappa}{\epsilon^2} \; log(\kappa / \epsilon)\right)$ & $\mathcal{O}\left(\left\lceil\dfrac{log(\epsilon)}{\log({\kappa \;\epsilon_{l}})}\right\rceil \dfrac{\mathcal{B}\kappa}{\epsilon_{l}^2} \; log(\kappa / \epsilon_{l})\right)$
    \end{tabular}
        \caption{Quantum cost for QSVT-based LS solution with and without iterative refinement.}

    \label{tab:complexity}
\end{table}

Using iterative refinement jointly with QSVT to solve linear systems provides some advantages. First, working with a lower precision decreases the number of samples needed to reach this precision. Actually, to get a precision $\epsilon$ we need $\mathcal{O}(1/\epsilon^2)$ samples, and then we need to run the quantum circuit $\mathcal{O}(1/\epsilon^2)$ times. 
Then, something more specific to the QSVT is that we can adapt the precision of the polynomial approximation of the inverse function to the accuracy $\epsilon_l$ used in the iteration. Reducing this precision also reduces the degree of the polynomial and the resulting number of calls to the block-encoding $U$ of $A^\dagger$ (and $U^\dagger$).
Finally we will verify in our experiments in Section~\ref{sec:exp} that, for our experimental values of $\ep$, $\ep_l$ and $\kappa$ (with $\epsilon < \epsilon_l < 1/\kappa$) and using less than $\lceil\log(\ep) / \log(\ep_l \ka)\rceil$ iterations, the quantum cost is smaller when using iterative refinement.

\subsubsection{Classical cost}\label{sec:discussion}
First, the QSVT (with or without iterative refinement) requires to pre-process (on a classical computer) the input matrix to create its block-encoding circuit \cite{HBG24,KBG24}. This task can be computationally expensive especially for dense and unstructured matrices (e.g., $\bigo (n4^n)$ flops using~\cite{HBG24}). In both columns of Table \ref{tab:complexity} this pre-processing step is performed only once (for the iterative approach the result is reused by the QPU throughout the iterations). Another initial computation consists in finding the phases $\Phi$ for the QSVT, using algorithm such as \cite{DLNW23, NJ24}. The computational cost of this task can scale linearly with the condition number $\ka$ (see \cite{NJ24}).
An additional cost concerns the state preparation of the right-hand sides with the generation of the corresponding circuits. For instance the algorithm provided in \cite{KP16} relies on a tree that needs to be classically computed and can be performed in $\bigo(N)$ flops.
Then, for the quantum iterative refinement method we need to perform some processing on the CPU before and after each solve. These tasks consist in normalizing the residual before the solve, then de-normalizing the result sent back from the QPU at each iteration (see Remark \ref{rem:normalization}) and finding the associated residual. The normalization can be performed in $\bigo(N)$ flops and finding the residual is a matrix-vector operation ($\bigo(N^2)$ flops). The de-normalization step can be performed for example using the Brent's method \cite{Brent73} (in the worst case scenario the complexity is $\bigo(log(1/\ep))$).

\subsubsection{Remarks on data communication}
Our algorithm involves data communication between the CPU and the QPU. At the beginning we send the circuit representation of the block-encoding of $A^\dagger$ (denoted as BE($A^\dagger$)) from the CPU to the QPU. This communication cost will depend on the matrix $A$ but also on the method used to perform the block-encoding and that determines the circuit's size (see Section \ref{sec:block_encoding}). This data transfer has to be performed only once because the matrix $A$ remains unchanged during the whole refinement process. Another transfer concerns the vector of phases $\Phi$ used in the QSVT circuit. This vector is of size $d$, where $d$ is the degree of the polynomial approximating the inverse function in Equation (\ref{eq:poly}).
Then at each solve step, we need to transfer the right-hand sides ($b$ and the residuals) from the CPU to the QPU. This is performed by sending the corresponding circuit description for state preparation (denoted as SP($b$), SP($r_i$), etc). The size of these circuits depends also on the method used to perform the state preparation. Moreover, after the solve phase, we need to transfer from the QPU to the CPU the sampled solution, which is a vector of size $N=2^n$. These data transfers between CPU and QPU are depicted in Figure \ref{fig:schema_comm}. 

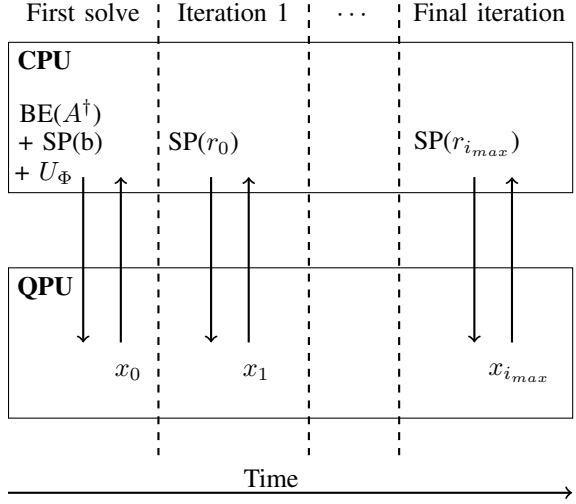
\begin{figure}[ht]
    \centering
    \begin{tikzpicture}
    \tikzstyle{rect} = [draw, minimum width=7.5cm, minimum height=2cm, anchor=north west]

    \node[rect] (cpu) at (0,0) {};
    \node[anchor=north west] at (cpu.north west) {\textbf{CPU}};

    \node[rect] (qpu) at (0,-3) {};
    \node[anchor=north west] at (qpu.north west) {\textbf{QPU}};

    \draw[->, thick] (1,-1.8) -- (1,-4);
    \node[] at (0.7,-0.95){BE($A^\dagger$)};
    \node[] at (0.7, -1.35) { + SP(b)};
    \node[] at (0.48, -1.75) {+ $U_\Phi$};
    \draw[->, thick] (1.5,-4) -- (1.5,-1.8);
    \node[] at (1.6,-4.4){$x_0$};

    \draw[dashed, thick] (2,0.5) -- (2,-5.5);
    \node[above] at (1,0.15) {First solve};

   \node at (0.5, -6.5) {~~};

    \draw[->, thick] (2.7,-1.8) -- (2.7,-4);
    \node[] at (2.6, -1.35) {SP($r_0$)};
    \draw[->, thick] (3.2,-4) -- (3.2,-1.8);
    \node[] at (3.3,-4.4){$x_1$};

    \draw[dashed, thick] (4,0.5) -- (4,-5.5);
    \node[above] at (3,0.15) {Iteration 1};

    \draw[dashed, thick] (5.2,0.5) -- (5.2,-5.5);
    \node[above] at (4.6,0.15) {$\cdots$};

    \draw[->, thick] (6.2,-1.8) -- (6.2,-4);
    \node[] at (6.1, -1.35) {SP($r_{i_{max}}$)};
    \draw[->, thick] (6.7,-4) -- (6.7,-1.8);
    \node[] at (6.8,-4.4){$x_{i_{max}}$};
    \node[above] at (6.4,0.15) {Final iteration};

    \draw[->, thick] (0,-6) -- (7.5,-6);
    \node[] at (3.5, -5.8) {Time};

\end{tikzpicture}
    \caption{CPU-QPU communication scheme for Algorithm \ref{alg:iter_ref_QSVT} (BE=block-encoding, SP=state preparation).}
    \label{fig:schema_comm}
\end{figure}


\subsubsection{Practical example}\label{sec:usecase}

Let us consider the one-dimensional Poisson equation
\begin{equation}\label{eq:poisson}
\forall x \in (0,1), -u''(x) = f(x),     
\end{equation}
with the Dirichlet boundary conditions
$
u(0) = u(1) = 0
$.
This problem can be solved using the finite difference method by discretizing Equation (\ref{eq:poisson}) with a step $h=1/(N+1)$ and $u_j=jh$, $f_j = f(jh)$. 

The solution can be obtained by solving the linear system provided in Equation \ref{eq:tridiag}.

\begin{equation}\label{eq:tridiag}
\frac{1}{h^2}
\begin{pmatrix}
    2 & -1 & & & 0\\
    -1 & 2 & -1 & & \\
    & -1 & \ddots& \ddots& \\
    & & \ddots & 2 &-1 \\
    0 & & & -1 & 2  
\end{pmatrix}
\begin{pmatrix}
    u_1\\
    \vdots\\
    u_N
\end{pmatrix}
=
\begin{pmatrix}
    f_1\\
    \vdots\\
    f_N
\end{pmatrix}.
\end{equation}

\paragraph{Block-encoding}

For the matrix given in Equation~(\ref{eq:tridiag}) (with $N=2^n$) we use the block-encoding technique provided in \cite{BE_EDP_2024} resulting in the circuit given in Figure \ref{fig:circuitBE}.

\begin{figure}[ht!]
    \centering
    \includegraphics[width = \linewidth]{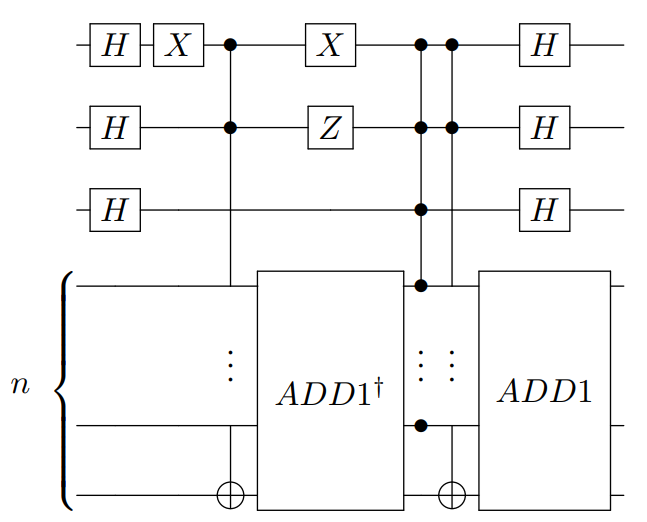}
    \caption{Circuit for the block-encoding of the tridiagonal matrix in Equation (\ref{eq:tridiag}).}
    \label{fig:circuitBE}
\end{figure}

\paragraph{Complexity}

In this section we detail the quantum and classical complexities of our hybrid algorithm when solving Equation (\ref{eq:tridiag}).
The classical cost is expressed in floating-point operations (flops) and the quantum cost is expressed in number of T-gates, because the depth of the circuit requires to use a fault-tolerant quantum computer~\cite{HFDM12}.
The subroutines used in our algorithm are: 
\begin{itemize}
    \item Tree-based state preparation (SP) described in \cite{KP16},
    \item Block-encoding (BE) provided in \cite{BE_EDP_2024},
    \item QSVT quantum circuit ($U_{\Phi}$...) proposed in \cite{GSLW19},
    \item QSVT phases ($\Phi$) computed as in algorithm given in \cite{NJ24},
    \item Solution via Brent's method \cite{Brent73} (de-normalization, see Remark \ref{rem:normalization}).
\end{itemize}

In Table \ref{tab:Poisson} we summarize the classical and quantum costs for the first solve (First) and at each iteration (Iter). The block-encoding of the tridiagonal matrix is predetermined (expressed analytically in \cite{BE_EDP_2024}) and requires no classical cost. We refer to \cite{KG24, RV25} for the decomposition in T gates of the multi-controlled Toffoli gates and adders for the different quantum circuits of our implementation.

\begin{table}[ht!]
\centering
\begin{tabular}{|l|l|c|c|}
\hline
 &  & Classical & Quantum\\ 
\hline

\multirow{4}{*}{First}
& SP & $\bigo(2^n)$ & $\bigo(polylog(n))$ \\
& BE  & - & $\bigo(n\ka \log(\ka /\ep_l))$\\
& QSVT ($\Phi$, $U_{\Phi}$...) & $\bigo(\ka)$ & $\bigo(n\ka \log(\ka /\ep_l))$ \\
& Solution & $\bigo(4^n + \log(1/\ep))$ & - \\
\hline
\hline

\multirow{4}{*}{Iter}
& SP & $\bigo(2^n)$ & $\bigo(polylog(n))$ \\
& BE & - & $\bigo(n\ka \log(\ka /\ep_l))$ \\
& QSVT ($\Phi$, $U_{\Phi}$...) & - & $\bigo(n\ka \log(\ka /\ep_l))$ \\
& Solution & $\bigo(4^n + \log(1/\ep))$& - \\

\hline
\end{tabular}
\caption{Complexity for solving the Poisson equation with mixed precision iterative refinement.}
\label{tab:Poisson}
\end{table}

We point out that this use case is given as an example to illustrate the complexity breakdown of our algorithm. Current classical solvers are efficient at solving this type of linear systems (in $\bigo(N)$ flops~\cite{BHM00}). Moreover the condition number $\ka$ exhibits a rapid increase with the problem size (in $\bigo(N^2)$ with no preconditioning~\cite{LM21}) which makes this linear system solution very expensive for large matrices using QSVT, given the current state of the art.

\section{Numerical experiments}\label{sec:exp}

\subsection{Experimental framework}

We have implemented the QSVT solver and the iterative refinement using Python and the \textit{myQLM} simulator \cite{myqlm}.
Our implementation relies on the following components:
\begin{itemize} 
    \item State preparation using the tree-based method described in \cite{KP16}, 
    \item Polynomial approximation of the inverse function as detailed in \cite{GSLW19},
    \item Algorithm for determining the QSVT angles: for small condition numbers ($\ka < 100$), we perform the numerical computation of the angles, using the symmetric QSP approach from \cite{DLNW23}. For bigger condition numbers, we use the estimation approach provided in \cite{NJ24},
    \item Implementation of the QSVT implemented with myQLM, following \cite{GSLW19},
    \item Iterative refinement in Python, using numpy and scipy libraries. 
\end{itemize}

In the following experiments, the size of the problem was set to $N = 16$ i.e., $n=4$ qubits. 
The matrix $A \in \mathbb{R}^{N \times N}$ and the vector $b \in \mathbb{R}^N$ are both randomly generated, with $\norme{b} = 1$ for simplicity. We have also simulated the algorithm with the tridiagonal matrix given in Section~\ref{sec:usecase}. The obtained results are similar in terms of convergence and then they are not mentioned in the following.
Note that for each simulation we limited the execution time to one hour, which enabled us to simulate problems with condition numbers of $\bigo(10^2)$.

\subsection{Results}
We present in Figure \ref{fig:iterations} the evolution of the scaled residual obtained at each iteration, until convergence. We consider a small condition number $\ka =10$ and three values for $\ep_l$. The targeted accuracy is $\ep=10^{-11}$.
We observe that in all cases the bound $\lceil\log(\ep) / \log(\ep_l \ka)\rceil$ obtained in Theorem \ref{theo:residual} is a sharp estimate for the iteration count. We also observe that due to the small values of $\ka$ the scaled residual is here also a good estimate of the forward error, as expected from Equation~(\ref{residual_error}).

\begin{figure}[htbp!]
        \centering
        \includegraphics[width=0.8\linewidth]{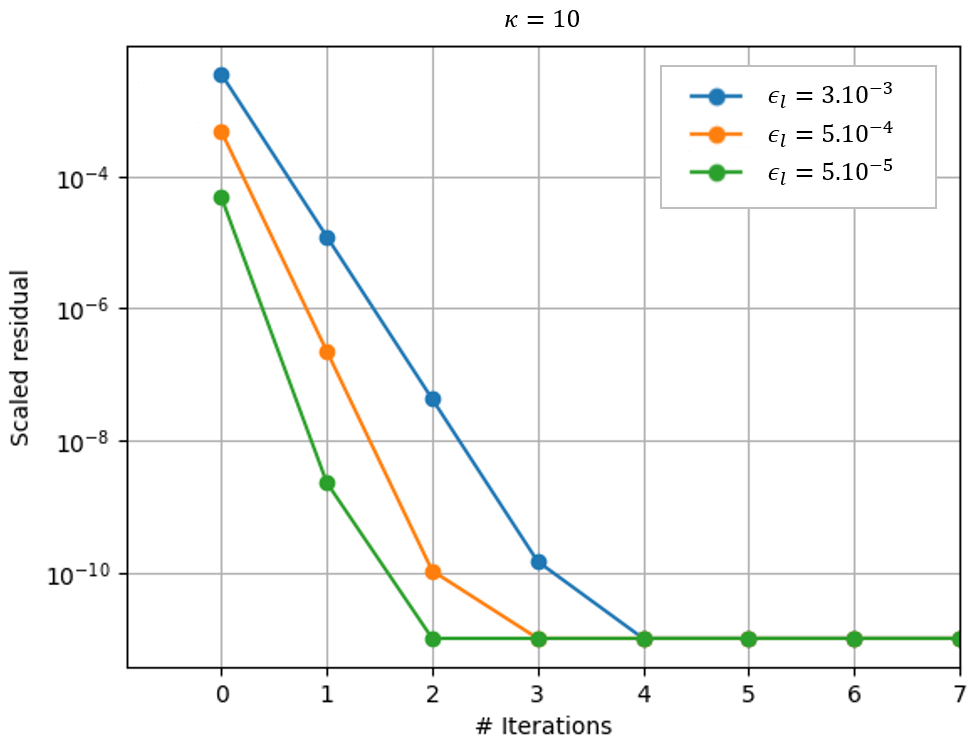}
    \caption{Scaled residual until convergence for $\ka = 10$, targeted accuracy $\ep=10^{-11}$, and various values of $\ep_l$.}
    \label{fig:iterations}
\end{figure}

In Figure \ref{fig:largekappa} we present the evolution of the scaled residual at each iteration for larger condition numbers. For these experiments, the polynomial approximation and the QSVT angles are computed by the algorithm implemented in \cite{NJ24}. In this context, we only have control on $\kappa$ since the value of $\epsilon_l$ is automatically determined by the algorithm from \cite{NJ24}. Our experiments show that we obtain a satisfying convergence with a number of iterations still lower than the bound from Theorem \ref{theo:residual}.

\begin{figure}
    \centering
    \includegraphics[width=0.8\linewidth]{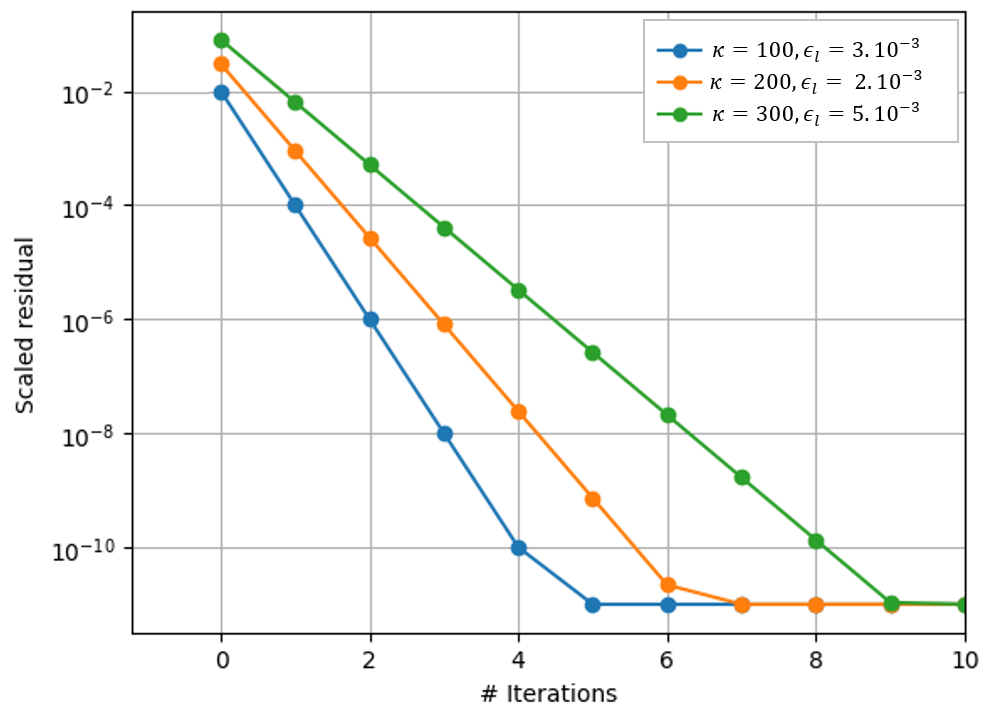}
    \caption{Scaled residual until convergence, $\kappa = 100,200,300$.}
    \label{fig:largekappa}
\end{figure}

\begin{figure}[htbp!]
    \centering
    \includegraphics[width=0.8\linewidth]{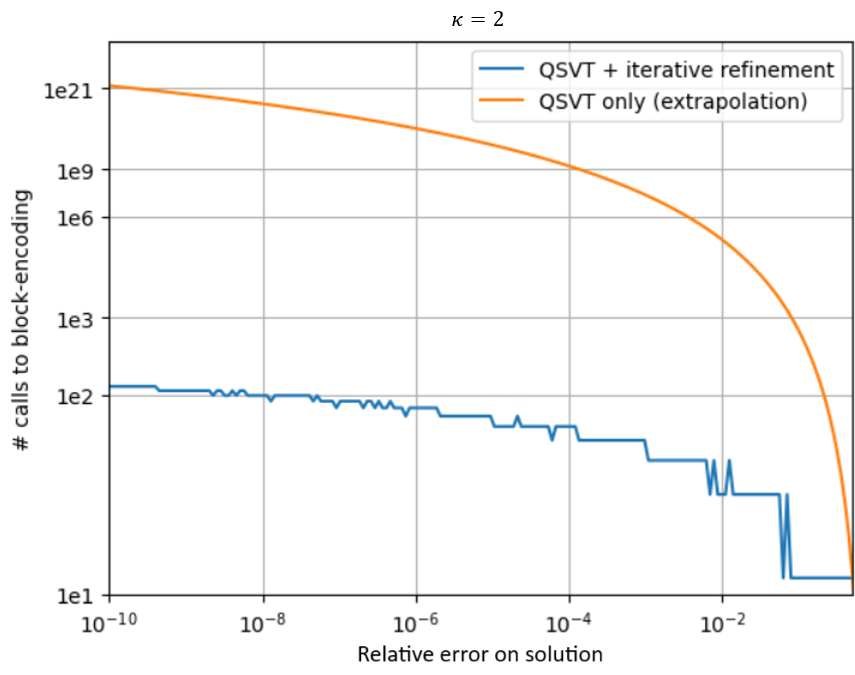}
    \caption{Complexity in calls to block-encoding for QSVT with and without iterative refinement, $\kappa=2$.}
    \label{fig:complexity}
\end{figure}

Then, in Figure \ref{fig:complexity}, we compare the complexity of the LS solving with $\kappa = 2$ for the QSVT with and without mixed-precision iterative refinement. In this numerical experiment we consider that the main cost of the QSVT relies in the block-encoding, so the complexity is expressed as the number of calls to the block-encoding of the matrix $A^\dagger$. The results for the QSVT without iterative refinement are extrapolated from the theoretical complexity provided in Table \ref{tab:complexity} since the computation would be intractable for both the polynomial approximation and the simulation of the quantum circuit. On the contrary, the results for the QSVT with mixed-precision iterative refinement are obtained after running our algorithm using a quantum simulator and choosing $\epsilon_l \approx 1/\kappa$. The method with iterative refinement is clearly advantageous for $\epsilon << \epsilon_l$. Note that the two curves coincide at $\epsilon = \epsilon_l$. With larger condition numbers, the gap between the two curves would be bigger, according the complexity figures given in Table \ref{tab:complexity}. 

\section{Conclusion}\label{sec:conclusion}

We have proposed a mixed-precision hybrid CPU/QPU solver for linear systems that computes a first solution in low precision using the (expensive) QSVT method and then refines the solution in higher precision to achieve the desired accuracy.
The main advantage of this method is that it enables us to achieve satisfying accuracy by using affordable quantum resources since we can use a limited precision for the QSVT solver. This solver illustrates what a typical hybrid CPU/QPU algorithm would be in the future that would leverage the strength of each architecture.

The interest of developing such methods is to exploit existing theoretical results for mixed-precision algorithms used in classical computing.
Our algorithm has been tested in simulation only which reduces the possibility of scaling to larger problems or handling ill-conditioned matrices. It requires further testing on real quantum machine when error-free hardware will be easily accessible to the scientific community.

\section*{Acknowledgement}
This work is part of the HQI initiative ({\tt www.hqi.fr})
and is supported by France 2030 under the French National
Research Agency award number “ANR-22-PNCQ-0002”.
\bibliographystyle{plain}
\bibliography{bib}

@misc{LINPACK,
 author    = {A. Petitet and R. C. Whaley and J. Dongarra and A. Cleary},
 title     = {{HPL} - A Portable Implementation of the High-Performance {L}inpack Benchmark for Distributed-Memory Computers},
 year       = {2018}
}

@book{LAPACK,
    author = {E. Anderson and Z. Bai and C. Bischof and S. Blackford and J. Demmel and J. Dongarra and J. Du Croz and A. Greenbaum and S. Hammarling and A. McKenney and D. Sorensen},
    title = {{LAPACK} Users' Guide},
    publisher = {SIAM},
    address   = {Philadelphia},
    year = 1999,
    note = {3rd edition}
}

@book{H96,
    author = "Nicholas J. Higham",
    title = "Accuracy and Stability of Numerical Algorithms",
    publisher = "SIAM",
    address =   "Philadelphia",
    year = 1996
}

@book{LAUB,
    author = "Alan J. Laub",
    title = "Computational matrix analysis",
    publisher = "SIAM",
    address =   "Philadelphia",
    year = 2012
}

@book{GVL13,
    author    = {G.~H. Golub and C.~F. {Van Loan}},
    title     = {Matrix Computations},
    publisher  = {The Johns Hopkins University Press},
    address   = {Baltimore},
    year       = {2013},
    note      = {Fourth edition}
}

@book{W63,
    title={Rounding Errors in Algebraic Processes},
    author={Wilkinson, J. H.},
    series={Prentice-Hall Series in Automatic Computation},
    url={https://books.google.fr/books?id=CcNBvwEACAAJ},
    year={1963},
    publisher={Prentice-Hall}
}

@article{Mixed_2016, 
    title={Linear algebra software for large-scale accelerated multicore computing}, 
    volume={25}, 
    DOI={10.1017/S0962492916000015}, 
    journal={Acta Numerica}, 
    author={Abdelfattah, A. and Anzt, H. and Dongarra, J. and Gates, M. and Haidar, A. and Kurzak, J. and Luszczek, P. and Tomov, S. and Yamazaki, I. and YarKhan, A.}, year={2016}, 
    pages={1–160}
}

@article{BBDK09,
    author = {Marc Baboulin and Alfredo Buttari and Jack Dongarra and Jakub Kurzak and Julie Langou and Julien Langou and Piotr Luszczek and Stanimire Tomov"},
    title = {Accelerating scientific computations with mixed precision algorithms},
    journal = {Computer Physics Communications},
    volume = {180},
    year = {2009},
    pages = {2526-2533}
}

@article{CH18,
    author = {Carson, Erin and Higham, Nicholas J.},
    title = {Accelerating the Solution of Linear Systems by Iterative Refinement in Three Precisions},
    journal = {SIAM Journal on Scientific Computing},
    volume = {40},
    number = {2},
    pages = {A817-A847},
    year = {2018}
}

@inproceedings{Baboulin_Europar,
    author       = {M. Baboulin and
                  S. Donfack and
                  O. Kaya and
                  T. Mary and
                  M. Robeyns},
    title        = {Mixed Precision Randomized Low-Rank Approximation with {GPU} Tensor Cores},
    booktitle    = {Euro-Par 2024: Parallel Processing - 30th European Conference on Parallel and Distributed Processing},
    series       = {Lecture Notes in Computer Science},
    volume       = {14803},
    pages        = {31--44},
    publisher    = {Springer},
    year         = {2024}
}

@article{HM22, 
    title={Mixed precision algorithms in numerical linear algebra},
    volume={31},
    DOI={10.1017/S0962492922000022},
    journal={Acta Numerica},
    author={Higham, Nicholas J. and Mary, Theo},
    year={2022},
    pages={347–414}}

@article{LM21,
    title = {Optimal preconditioners on solving the {P}oisson equation with {N}eumann boundary conditions},
    journal = {Journal of Computational Physics},
    volume = {433},
    pages = {110189},
    year = {2021},
    issn = {0021-9991},
    doi = {https://doi.org/10.1016/j.jcp.2021.110189},
    url = {https://www.sciencedirect.com/science/article/pii/S002199912100084X},
    author = {Byungjoon Lee and Chohong Min}
}

@book{BHM00,
    author = {Briggs, William and Henson, Van and McCormick, Steve},
    year = {2000},
    publisher = {SIAM: Society for Industrial and Applied Mathematics; 2nd edition},
    title = {A Multigrid Tutorial, 2nd Edition}
}

@book{Brent73,
    author="Brent, R. P.",
    title = "Algorithms for Minimization without Derivatives",
    publisher= "Prentice-Hall",
    address = "Englewood Cliffs, NJ",
    year = "1973",
}

@Book{R20,
    author    = "Rivlin, T. J.",
    title     = "Chebyshev Polynomials: From Approximation Theory to Algebra and Number Theory.",
    publisher = "Dover",
    year      = "2020"
}

@inproceedings{BSH20,
    doi = {10.23919/date48585.2020.9116502},
    url = {https://doi.org/10.23919/date48585.2020.9116502},
    year = {2020},
    publisher = {{IEEE}},
    author = {K. Bertels and A. Sarkar and T. Hubregtsen. and M. Serrao and A. A. Mouedenne and A. Yadav and A. Krol and I. Ashraf},
    title = {Quantum Computer Architecture: Towards Full-Stack Quantum Accelerators},
    booktitle = {2020 Design,  Automation {\&} Test in Europe Conference {\&} Exhibition (2020)}
}

@article{HMLG21,
    author={Humble, Travis S. and McCaskey, Alexander and Lyakh, Dmitry I. and Gowrishankar, Meenambika and Frisch, Albert and Monz, Thomas},
    journal={IEEE Micro}, 
    title={Quantum Computers for High-Performance Computing}, 
    year={2021},
    volume={41},
    number={5},
    pages={15-23},
    doi={10.1109/MM.2021.3099140}
}

@article{HFDM12,
    title={Surface code quantum computing by lattice surgery},
    volume={14},
    ISSN={1367-2630},
    url={http://dx.doi.org/10.1088/1367-2630/14/12/123011},
    DOI={10.1088/1367-2630/14/12/123011},
    number={12},
    journal={New Journal of Physics},
    publisher={IOP Publishing},
    author={Horsman, Dominic and Fowler, Austin G and Devitt, Simon and Meter, Rodney Van},
    year={2012},
    pages={123011} 
}

@article{RV25,
    title={Ancilla-free Quantum Adder with Sublinear Depth}, 
    author={Maxime Remaud and Vivien Vandaele},
    year={2025},
    journal = {arXiv:2501.16802}
}

@article{KG24,
    title={Rise of conditionally clean ancillae for optimizing quantum circuits}, 
    author={Tanuj Khattar and Craig Gidney},
    year={2024},
    journal={arXiv:2407.17966} 
}

@article{HBG24,
    doi = {10.1088/1402-4896/ad6499},
    url = {https://dx.doi.org/10.1088/1402-4896/ad6499},
    year = {2024},
    month = {jul},
    publisher = {IOP Publishing},
    volume = {99},
    number = {8},
    pages = {085128},
    author = {Hantzko, Lukas and Binkowski, Lennart and Gupta, Sabhyata},
    title = {Tensorized Pauli decomposition algorithm},
    journal = {Physica Scripta},
}

@INPROCEEDINGS{KBG24,
    author={Koska, Océane and Baboulin, Marc and Gazda, Arnaud},
    booktitle={ISC High Performance 2024 Research Paper Proceedings (39th International Conference)}, 
    title={A Tree-Approach Pauli Decomposition Algorithm with Application to Quantum Computing}, 
    year={2024},
    pages={1-11},
    doi={10.23919/ISC.2024.10528938}
}

@InProceedings{KP16,
    author =	{Kerenidis, Iordanis and Prakash, Anupam},
    title =	{{Quantum Recommendation Systems}},
    booktitle =	{8th Innovations in Theoretical Computer Science Conference (ITCS 2017)},
    pages =	{49:1--49:21},
    series =	{Leibniz International Proceedings in Informatics (LIPIcs)},
    ISBN =	{978-3-95977-029-3},
    ISSN =	{1868-8969},
    year =	{2017},
    volume =	{67}
}

@article{BE_EDP_2024,
    author={Sunheang Ty and Renaud Vilmart and Axel TahmasebiMoradi and Chetra Mang},
    title={Double-Logarithmic Depth Block-Encodings of Simple Finite Difference Method's Matrices}, 
    journal = {arXiv:2410.05241},
    year = {2024}
}

@article{CLVY23,
    author = {Camps, Daan and Lin, Lin and Van Beeumen, Roel and Yang, Chao},
    title = {Explicit Quantum Circuits for Block Encodings of Certain Sparse Matrices},
    journal = {SIAM Journal on Matrix Analysis and Applications},
    volume = {45},
    number = {1},
    pages = {801-827},
    year = {2024},
    doi = {10.1137/22M1484298}
}

@INPROCEEDINGS{CV22,
    author={Camps, Daan and Van Beeumen, Roel},
    booktitle={2022 IEEE International Conference on Quantum Computing and Engineering (QCE)}, 
    title={{FABLE}: Fast Approximate Quantum Circuits for Block-Encodings}, 
    year={2022},
    volume={},
    number={},
    pages={104-113}
}

@article{CW12, 
    author = {A. M. Childs and N. Wiebe},
    title = {Hamiltonian simulation using linear combinations of unitary operations},
    journal={Quantum Information and Computation},
    volume={12},                 
    number={11 - 12},
    pages={0901-0924},
    publisher={Rinton Press},
    year={2012}
}

@article{DLNW23,
    author = {Dong, Yulong and Lin, Lin and Ni, Hongkang and Wang, Jiasu},
    title = {Robust Iterative Method for Symmetric Quantum Signal Processing in All Parameter Regimes},
    journal = {SIAM Journal on Scientific Computing},
    volume = {46},
    number = {5},
    pages = {A2951-A2971},
    year = {2024},
    doi = {10.1137/23M1598192}
}

@article{LYC16,
    title={Methodology of Resonant Equiangular Composite Quantum Gates},
    volume={6},
    number={4},
    journal={Physical Review X},
    publisher={American Physical Society (APS)},
    author={Low, Guang Hao and Yoder, Theodore J. and Chuang, Isaac L.},
    year={2016}
}

@inproceedings{GSLW19, 
    series={STOC ’19},
    title={Quantum singular value transformation and beyond: exponential improvements for quantum matrix arithmetics},
    url={http://dx.doi.org/10.1145/3313276.3316366},
    DOI={10.1145/3313276.3316366},
    booktitle={Proceedings of the 51st Annual ACM SIGACT Symposium on Theory of Computing},
    publisher={ACM},
    author={Gilyén, András and Su, Yuan and Low, Guang Hao and Wiebe, Nathan},
    year={2019},
    collection={STOC ’19}
}

@article{MRTC21,
    title={Grand Unification of Quantum Algorithms},
    volume={2},
    ISSN={2691-3399},
    url={http://dx.doi.org/10.1103/PRXQuantum.2.040203},
    DOI={10.1103/prxquantum.2.040203},
    number={4},
    journal={PRX Quantum},
    publisher={American Physical Society (APS)},
    author={Martyn, John M. and Rossi, Zane M. and Tan, Andrew K. and Chuang, Isaac L.},
    year={2021}
}

@article{NJ24,
    title = {Estimating {QSVT} angles for matrix inversion with large condition numbers},
    journal = {Journal of Computational Physics},
    volume = {525},
    pages = {113767},
    year = {2025},
    issn = {0021-9991},
    doi = {https://doi.org/10.1016/j.jcp.2025.113767},
    author = {I. Novikau and I. Joseph}
}

@article{L22,
    author = {Lin Lin},
    title = {Lecture notes on quantum algorithms for scientific computation} ,
    journal = {arXiv:2201.08309},
    year = {2022}
}

@article{VQLS,
    title={Variational Quantum Linear Solver},
    volume={7},
    journal={Quantum},
    publisher={Verein zur Forderung des Open Access Publizierens in den Quantenwissenschaften},
    author={Bravo-Prieto, Carlos and LaRose, Ryan and Cerezo, M. and Subasi, Yigit and Cincio, Lukasz and Coles, Patrick J.},
    year={2023},
    pages={1188}
}

@article{IR_HHL_zheng2024,
    title={An Early Investigation of the {HHL} Quantum Linear Solver for Scientific Applications}, 
    author={Muqing Zheng and Chenxu Liu and Samuel Stein and Xiangyu Li and Johannes Mülmenstädt and Yousu Chen and Ang Li},
    journal={arXiv:2404.19067},
    year={2024}
}

@article{IR_HHL_saito2021,
    title={An Iterative Improvement Method for {HHL} algorithm for Solving Linear System of Equations}, 
    author={Yoshiyuki Saito and Xinwei Lee and Dongsheng Cai and Nobuyoshi Asai},
    journal={arXiv:2108.07744},
    year={2021}
}

@article{HHL09,
    title={Quantum Algorithm for Linear Systems of Equations},
    volume={103},
    ISSN={1079-7114},
    url={http://dx.doi.org/10.1103/PhysRevLett.103.150502},
    DOI={10.1103/physrevlett.103.150502},
    number={15},
    journal={Physical Review Letters},
    publisher={American Physical Society (APS)},
    author={Harrow, Aram W. and Hassidim, Avinatan and Lloyd, Seth},
    year={2009}
}

@article{IR_opti,
    author = {Mohammadisiahroudi, Mohammadhossein and Augustino, Brandon and Sampourmahani, Pouya and Terlaky, Tamás},
    year = {2025},
    month = {01},
    pages = {},
    title = {Quantum computing inspired iterative refinement for semidefinite optimization},
    journal = {Mathematical Programming},
    doi = {10.1007/s10107-024-02183-z}
}

@misc{myqlm,
    author = {{Eviden Quantum Lab}},
    title = {{myQLM: Quantum Computing Framework}},
    year = {2020-2024},
    note = {https://myqlm.github.io}
}

\end{document}